\documentclass[
  aps,
  prd,
  preprint,
  notitlepage,
  superscriptaddress,
  nofootinbib,
  longbibliography
]{revtex4-2}

\usepackage{amsmath,amssymb}
\usepackage{booktabs}
\usepackage{graphicx}
\usepackage{xcolor}

\newtheorem{theorem}{Theorem}
\newtheorem{corollary}{Corollary}
\newtheorem{conjecture}{Conjecture}
\newtheorem{remark}{Remark}
\newtheorem{problem}{Problem}
\newenvironment{proof}{\par\noindent\textit{Proof.}\ }{\hfill$\square$\par}

\newcommand{\amu}{m_{\mathrm u}}
\newcommand{\dx}{\Delta x}
\newcommand{\sinc}{\operatorname{sinc}}
\newcommand{\Lcsl}{\Lambda_{\mathrm{CSL}}}
\newcommand{\Ldp}{\Lambda_{\mathrm{DP}}}
\newcommand{\Krc}{K_{r_C}}
\newcommand{\Xic}{\Xi_{\mathrm{CSL/DP}}}
\newcommand{\ellstar}{\ell_*}

\begin{document}

\title{Geometry-Only CSL/DP Ratios and the Nonuniqueness of Decoherence Kernels}

\author{Randy Davila}
\email{randy@firstprinciples.com}
\affiliation{FirstPrinciples, Toronto, Ontario, Canada}
\affiliation{Department of Computational Applied Mathematics and Operations
Research, Rice University, Houston, Texas 77005, USA}

\author{Gerard J. Milburn}
\email{gerard.milburn@stfc.ac.uk}
\affiliation{School of Mathematics and Physics, The University of Sussex, UK}
\affiliation{National Quantum Computing Centre,
Rutherford Appleton Laboratory, UK}

\date{August 4, 2026}

\begin{abstract}
We study idealized levitated protocols that create spatial superpositions of massive test particles. For each protocol, we compare the dimensionless contrast-loss exponent of mass-proportional continuous spontaneous localization (CSL) with the Di\'osi--Penrose (DP) self-energy exponent $E_G\tau/\hbar$. We first prove that the point-particle CSL separation kernel has an exact random-unitary realization: Gaussian momentum kicks arriving at Poisson-distributed times produce the same unconditional decay of spatial coherence, although a pure state conditioned on the complete kick record remains pure. The separation kernel alone therefore specifies an operational decoherence law, not the occurrence of objective collapse. We then show that the ratio of the CSL and DP exponents is independent of particle mass and interrogation time. In the point-particle model it depends only on branch separation and an effective distance; for the standard GRW reference parameters, its resolved-superposition crossover is $x_*\approx1.91\,\mathrm{nm}$. For rigid spherical bodies with an arbitrary normalized radial mass profile, total mass, overall density scale, and interrogation time again cancel, leaving a dimensionless geometry factor. The results distinguish three requirements for a decisive experiment: detectable absolute effects, a controlled comparison of CSL and DP scales, and observables capable of discriminating physically different dynamics that share the same ensemble decoherence kernel.
\end{abstract}


\maketitle

\section{Introduction}
\label{sec:introduction}
\emph{Levitated optomechanics} provides a controlled setting in which to test quantum mechanics with increasingly massive systems. A nanoparticle can be trapped, cooled, released, interrogated, and read using optical, magnetic, or hybrid control, while its isolation permits long-lived motional coherence. Reviews describe both the experimental control achieved in these platforms and their role in proposals for high-mass quantum superpositions~\cite{Aspelmeyer2014,Millen2020}. Mirror-superposition proposals and space-based concepts such as MAQRO further illustrate the foundational reach of this program~\cite{Marshall2003,Kaltenbaek2012,Kaltenbaek2015}. One motivation is to test dynamical-collapse models and related measures of macroscopic quantumness~\cite{Pearle1989,GhirardiPearleRimini1990,BassiGhirardi2003,Bassi2013,Nimmrichter2013,RomeroIsart2011}. Another is the Di\'osi--Penrose (DP) proposal, which associates the instability of a spatial superposition with the gravitational self-energy of the difference between its branch mass densities. Quantitative DP predictions depend on the chosen regularization and, in dissipative extensions, on additional dynamical assumptions~\cite{Diosi1987,Diosi1989,Penrose1996,Bahrami2014,Diosi2022}. Levitated and non-interferometric experiments continue to constrain such phenomenological modifications~\cite{Carlesso2022,DiBartolomeoCarlesso2024}.

Interpreting these experiments requires two distinctions. First, preparing and certifying a nonclassical mechanical state is not the same task as testing a modification of quantum dynamics. A protocol optimized for optomechanical correlations, for example, need not be optimal for detecting collapse-induced contrast loss~\cite{Pitchford2020}. Second, an unconditional master equation does not uniquely determine the stochastic evolution assigned to individual runs. Such a trajectory-level representation is commonly called an \emph{unravelling}. The same ensemble decoherence can result from inequivalent unravellings, including averages over unitary trajectories driven by unobserved classical noise~\cite{WisemanMilburn2010}. Ensemble visibility loss therefore does not, by itself, demonstrate that a wave function has objectively collapsed.

Against this background, we ask a specific design question. Consider the same idealized spatial superposition evaluated using the mass-proportional CSL contrast-loss exponent and the dimensionless DP self-energy exponent $E_G\tau/\hbar$. Which experimental parameters determine their ratio? Mass, branch separation, interrogation time, particle radius, and internal mass distribution all appear relevant at first sight. The conversion theorems below show that mass and time control the absolute magnitude of both effects but cancel from their ratio. Relative CSL/DP ordering is instead determined by the geometry of the object and of the superposition.

For a point particle, the CSL exponent contains the separation kernel
\[
\Krc(\dx)=1-\exp\!\left(-\frac{\dx^2}{4r_C^2}\right).
\]
Comparing this kernel with a regularized point-particle DP exponent gives a unique resolved-superposition crossover. For the conventional Ghirardi--Rimini--Weber (GRW) reference values $\lambda=10^{-17}\,\mathrm{s}^{-1}$ and $r_C=10^{-7}\,\mathrm m$, it occurs at $x_*\approx1.91\,\mathrm{nm}$. The cancellation of mass and time is not an artifact of the point-particle approximation: it survives when both quantities are written using finite-size kernels for any fixed normalized spherical mass profile.

The comparison does not identify CSL with DP; the two proposals have different physical motivations and dynamical content. It places their dimensionless instability exponents on the same protocol geometry and asks which is larger. Nor does observation of the CSL separation kernel uniquely identify CSL. We prove this by constructing a state-independent random-unitary process with exactly the same point-particle kernel. A visibility measurement can therefore constrain an unconditional generator while leaving its physical unravelling undetermined.

All DP statements are relative to the self-energy proxy and kernel normalization stated explicitly below. A different regularization changes the geometry factor and numerical crossover. It does not change the algebraic cancellation whenever the CSL and DP exponents retain the same $m^2\tau$ scaling.

The paper has three connected contributions. Section~\ref{sec:random-unitary} proves the random-unitary realization and establishes the interpretive limitation of the point-particle kernel. The following sections derive the point-particle crossover and the finite spherical-profile conversion theorem. The final analysis isolates radial mass placement as the remaining design variable and formulates a core-shell inversion conjecture supported by direct kernel evaluation.

\section{Model and Notation}
\label{sec:model_and_notation}
A protocol $P$ consists of a rigid spherical test mass, a spatial-superposition sequence, and an interrogation time. We first distinguish the point-particle contrast-loss kernel from a particular collapse unravelling. We then compare the point-particle CSL and DP scales with a radius-scale regularization, before replacing that comparison by finite-size kernels for a normalized spherical mass profile.

The model is deliberately minimal. It treats the branch separation $\dx$, the particle radius $R$, and the interrogation time $\tau$ as input parameters rather than deriving them from a trap, pulse sequence, or feedback protocol. This keeps the comparison independent of a particular platform. Experimental feasibility enters only later, through the question of which regions of $(R,\dx,\tau)$ are realistic for a given architecture.

\begin{itemize}
\item $m$ is the total particle mass and $\amu$ is the atomic mass unit.
\item $R$ is the particle radius.
\item $\dx>0$ is the separation between the two wave-packet centers.
\item $\tau$ is the interrogation or coherence time.
\item $\lambda$ and $r_C$ are the CSL collapse rate and localization length.
\item $G$ is Newton's gravitational constant and $\hbar$ is the reduced Planck constant.
\end{itemize}

The CSL separation kernel is
\[
\Krc(\dx)=1-\exp\!\left(-\frac{\dx^2}{4r_C^2}\right).
\]
It interpolates between the quadratic small-separation regime $\Krc(\dx)\sim\dx^2/(4r_C^2)$ and the resolved regime $\Krc(\dx)\to1$.

For the point-particle comparison we use the effective distance
\[
d_{\mathrm{eff}}=\max\{\dx,2R\}.
\]
The branch $d_{\mathrm{eff}}=\dx$ describes resolved superpositions in which the branch separation is at least the particle diameter. The branch $d_{\mathrm{eff}}=2R$ is a radius-scale regularization of the gravitational comparison when the branch centers are closer than one diameter.

The two dimensionless quantities compared below are $\Lcsl$ and $\Ldp$. The first is the CSL contrast-loss exponent. The second is the DP self-energy exponent $E_G\tau/\hbar$, interpreted as the elapsed interrogation time in units of the characteristic DP time $\hbar/E_G$. Thus $\Lcsl\gtrsim1$ and $\Ldp\gtrsim1$ each indicate an order-one effect within the corresponding phenomenological model. Their ratio is denoted by
\[
\Xic=\frac{\Lcsl}{\Ldp}.
\]

In the idealized CSL model, $e^{-\Lcsl}$ is the visibility multiplier for an off-diagonal coherence between the two branches. The quantity $\Ldp$ is not an ordinary unitary gravitational phase: it is the dimensionless self-energy scale used to define the characteristic DP reduction time. The ratio $\Xic$ therefore compares two phenomenological instability exponents assigned to the same branch geometry, without asserting that the underlying mechanisms are identical.

The figures use representative levitated-mass protocols, including free-expansion, matter-wave, collapse-sensitivity, optomechanical-correlation, and force-sensing examples. They illustrate the thresholds; the equalities and crossover statements are analytic.

There are consequently three separate questions. First, are $\Lcsl$ or $\Ldp$ large enough to be detected after preparation error, environmental decoherence, and readout noise are included? Second, which exponent is larger for the same idealized geometry? Third, if a particular contrast-loss law is observed, which physical dynamics generated it? The conversion theorems address the second question. Section~\ref{sec:random-unitary} shows why the third requires information beyond ensemble visibility.

\subsection{AI-assisted discovery and verification}

During the exploratory phase of this work, the first author used
\textsc{Theo-Conjecture} (version 0.0.1, development snapshot of July--August
2026) and OpenAI Codex (GPT-5, accessed July--August 2026). The
advisor-supervised loop organized parameterized model instances in an
inspectable registry, proposed candidate relations, and supported
counterexample searches, algebraic checks, literature organization, and
language revision. The first author directed these uses and checked the
retained outputs against explicit derivations, direct numerical evaluation,
or the cited literature. Both authors reviewed the resulting scientific claims
and take responsibility for the content of the manuscript.

\section{A CSL-Shaped Kernel Without Objective Collapse}
\label{sec:random-unitary}

We begin by isolating the operational content of the point-particle CSL kernel. Consider one translational degree of freedom with density operator $\rho$ and Hamiltonian $H$. Momentum kicks $\theta$, with dimensions of momentum, are drawn independently from a Gaussian distribution with mean zero and variance $\sigma_p^2$. The kicks arrive according to a Poisson process of rate $\gamma_m$, so that the probability of one kick in an interval $dt$ is $\gamma_m dt+o(dt)$. A kick acts through the unitary
\[
U_\theta=\exp\!\left(-\frac{i}{\hbar}\theta\hat x\right),
\]
and the corresponding Gaussian random-unitary channel is
\[
\mathcal E_{\sigma_p}[\rho]
=\int_{-\infty}^{\infty}
\frac{e^{-\theta^2/(2\sigma_p^2)}}{\sqrt{2\pi\sigma_p^2}}
U_\theta\rho U_\theta^\dagger\,d\theta.
\]
After averaging over both the kick values and the Poisson record, the ensemble state obeys
\begin{equation}
\frac{d\rho}{dt}
=-\frac{i}{\hbar}[H,\rho]
+\gamma_m\bigl(\mathcal E_{\sigma_p}[\rho]-\rho\bigr).
\label{eq:random-unitary-master}
\end{equation}

\begin{theorem}
Let $r_C=\hbar/(\sqrt{2}\sigma_p)$ and $\gamma_m=\lambda(m/\amu)^2$. Then the dissipative part of \eqref{eq:random-unitary-master} has position-basis matrix elements
\[
\left\langle x'\middle|
\gamma_m\bigl(\mathcal E_{\sigma_p}[\rho]-\rho\bigr)
\middle|x\right\rangle
=-\lambda\left(\frac{m}{\amu}\right)^2
\left[1-e^{-(x'-x)^2/(4r_C^2)}\right]\rho(x',x).
\]
Thus it reproduces exactly the point-particle CSL contrast-loss kernel used in this paper. Nevertheless, conditioned on the complete list of kick times and kick values, every trajectory is unitary and a pure initial state remains pure. The kernel therefore does not, by itself, distinguish objective collapse from classical random-unitary decoherence.
\end{theorem}

\begin{proof}
For each pair $x,x'$, the Gaussian characteristic function gives
\[
\left\langle x'\middle|\mathcal E_{\sigma_p}[\rho]\middle|x\right\rangle
=\int_{-\infty}^{\infty}
\frac{e^{-\theta^2/(2\sigma_p^2)}}{\sqrt{2\pi\sigma_p^2}}
e^{-i\theta(x'-x)/\hbar}\,d\theta\;\rho(x',x)
=e^{-\sigma_p^2(x'-x)^2/(2\hbar^2)}\rho(x',x).
\]
The identifications $r_C=\hbar/(\sqrt{2}\sigma_p)$ and $\gamma_m=\lambda(m/\amu)^2$ yield the stated kernel. For a realized record $\{(t_j,\theta_j)\}_{j=1}^N$, the state is obtained by alternating Hamiltonian evolution with the unitaries $U_{\theta_j}$. Their product is unitary, so it preserves purity. Finally, the Poisson rate is state-independent; the record probabilities therefore contain no state-dependent measurement likelihood that could produce localization by Bayesian conditioning.
\end{proof}

The same construction fixes the momentum diffusion contributed by the kicks. Since $U_\theta^\dagger\hat p U_\theta=\hat p-\theta$, averaging over one kick adds $\sigma_p^2$ to $\langle p^2\rangle$. Hence,
\begin{equation}
\left.\frac{d}{dt}\langle p^2\rangle\right|_{\mathrm{kick}}
=\gamma_m\sigma_p^2
=\lambda\left(\frac{m}{\amu}\right)^2\frac{\hbar^2}{2r_C^2}.
\label{eq:momentum-diffusion}
\end{equation}
For a free particle, this is the complete derivative because the Hamiltonian conserves $p^2$; for a general $H$, it is only the kick-induced contribution. The construction therefore has operational consequences even though it does not specify objective collapse. However, it is not a many-body derivation of mass-proportional CSL. The rate $\gamma_m$ has been chosen to match the fixed-mass point-particle exponent, whereas genuine CSL supplies a stochastic localization law and a prescribed amplification mechanism. The theorem establishes the nonuniqueness of the observed kernel, not equivalence of the underlying theories.

\section{CSL--DP Conversion and Crossover Theorems}
\label{sec:conversion_and_crossover}
We use a point-particle CSL contrast-loss exponent with collapse rate $\lambda$ and localization length $r_C$:
\[
\Lcsl(P)=\lambda\left(\frac{m}{\amu}\right)^2\tau
\left(1-\exp\!\left(-\frac{\dx^2}{4r_C^2}\right)\right).
\]
This is the standard separation dependence for suppression of spatial coherences in the point-particle approximation.

For the gravitational comparison, define the regularized distance
\[
d_{\mathrm{eff}}=\max\{\dx,2R\}
\]
and the dimensionless DP exponent
\[
\Ldp(P)=\frac{Gm^2\tau}{\hbar d_{\mathrm{eff}}}.
\]
This expression uses $Gm^2/d_{\mathrm{eff}}$ as a transparent point-particle proxy for the DP self-energy while avoiding a singularity at zero separation. It is not the exact DP self-energy of an extended body. Such a calculation requires the full mass distribution and a regularization prescription, which are introduced in the finite-profile section.

The effective distance should be read as a point-particle regularization, not as a claim about the exact self-energy of an extended particle. The finite-profile analysis below replaces this proxy by kernels for spherical mass distributions; the cancellation persists in that setting.

The reason this simplified comparison is still useful is that it exposes the scaling before detailed geometry is introduced. If mass or time failed to cancel already at this level, then finite-size corrections would not change the basic message. Instead, the point-particle calculation shows that the relative scale is already controlled by geometry, motivating the more faithful finite-size calculation below.

The quantity $\Ldp$ uses the positive DP self-energy associated with the difference between the two branch mass densities. It is therefore insensitive to the sign convention for the Newtonian potential energy and should not be confused with the unitary Newtonian phase accumulated by either branch.

\begin{theorem}[Point-particle conversion]
For every protocol $P$ in the model above,
\[
\Xic(P):=\frac{\Lcsl(P)}{\Ldp(P)}
=\frac{\lambda\hbar d_{\mathrm{eff}}}{G\amu^2}
\left(1-\exp\!\left(-\frac{\dx^2}{4r_C^2}\right)\right).
\]
In particular, at fixed geometry $(R,\dx)$, the ratio is independent of the particle mass $m$ and the interrogation time $\tau$.
\end{theorem}

\begin{proof}
By definition,
\[
\Lcsl=\lambda\frac{m^2}{\amu^2}\tau\Krc(\dx),
\qquad
\Ldp=\frac{Gm^2\tau}{\hbar d_{\mathrm{eff}}}.
\]
Dividing cancels $m^2$ and $\tau$, giving
\[
\frac{\Lcsl}{\Ldp}=\frac{\lambda\hbar d_{\mathrm{eff}}}{G\amu^2}\Krc(\dx).
\]
\end{proof}

\begin{corollary}[Small- and large-separation regimes]
If $\dx\ll r_C$, then
\[
\Xic(P)=\frac{\lambda\hbar d_{\mathrm{eff}}}{G\amu^2}
\left[\frac{\dx^2}{4r_C^2}
+O\!\left(\frac{\dx^4}{r_C^4}\right)\right].
\]
Thus, in the resolved-separation regime $\dx\ge2R$, the ratio is cubic to leading order in $\dx$; in the radius-regularized regime $\dx<2R$, it is quadratic in $\dx$ and linear in $R$. If $\dx\gg r_C$, then
\[
\Xic(P)\sim\frac{\lambda\hbar d_{\mathrm{eff}}}{G\amu^2},
\]
because the CSL separation kernel approaches unity.
\end{corollary}

\begin{proof}
The expansion follows from
\[
1-e^{-\dx^2/(4r_C^2)}=\frac{\dx^2}{4r_C^2}+O(\dx^4).
\]
The large-separation asymptotic follows because the exponential term tends to zero.
\end{proof}

\begin{remark}
Mass and interrogation time remain essential to detectability: they increase both $\Lcsl$ and $\Ldp$. They do not change the ratio because the two exponents have the same quadratic dependence on mass and linear dependence on time. Within this model, changing the relative ordering therefore requires changing the separation, radius, normalized mass profile, or collapse parameters rather than merely increasing $m$ or $\tau$.
\end{remark}

Define the crossover length
\[
\ellstar=\frac{G\amu^2}{\lambda\hbar}.
\]
Then $\Xic\ge1$ if and only if
\[
d_{\mathrm{eff}}\Krc(\dx)\ge\ellstar.
\]
Thus the comparison is governed by the single geometry threshold
\[
d_{\mathrm{eff}}\Krc(\dx)=\ellstar.
\]

\begin{theorem}[Resolved-superposition crossover]
Assume $d_{\mathrm{eff}}=\dx$, equivalently $\dx\ge2R$. Then there is a unique number $x_*>0$ such that
\[
x_*\left(1-\exp\!\left(-\frac{x_*^2}{4r_C^2}\right)\right)=\ellstar.
\]
Moreover,
\[
\Lcsl(P)\ge\Ldp(P)\quad\Longleftrightarrow\quad\dx\ge x_*.
\]
For the GRW reference values $\lambda=10^{-17}\,\mathrm{s}^{-1}$ and $r_C=10^{-7}\,\mathrm m$,
\[
\ellstar=1.745129893\times10^{-13}\,\mathrm m,
\qquad
x_*=1.911184110\,\mathrm{nm}.
\]
\end{theorem}

\begin{proof}
Let
\[
f(x)=x\left(1-e^{-x^2/(4r_C^2)}\right).
\]
For $x>0$,
\[
f'(x)=1-e^{-x^2/(4r_C^2)}+\frac{x^2}{2r_C^2}e^{-x^2/(4r_C^2)}>0.
\]
Also $f(0)=0$ and $f(x)\to\infty$ as $x\to\infty$. Hence there is a unique $x_*>0$ with $f(x_*)=\ellstar$. In the resolved regime,
\[
\Xic=\frac{\dx\Krc(\dx)}{\ellstar}=\frac{f(\dx)}{\ellstar},
\]
so $\Xic\ge1$ if and only if $\dx\ge x_*$. The numerical value follows by one-dimensional bisection.
\end{proof}

\begin{corollary}[Nanometer-scale rule of thumb]
Under the same GRW parameters and resolved-superposition assumption, any protocol with $\dx\ge2\,\mathrm{nm}$ has $\Lcsl>\Ldp$ in this model. Conversely, resolved protocols below $x_*\approx1.91\,\mathrm{nm}$ satisfy $\Lcsl<\Ldp$.
\end{corollary}

\begin{proof}
This follows immediately from the previous theorem and $2\,\mathrm{nm}>x_*$.
\end{proof}

\begin{corollary}[Resolved-crossover scaling]
When $x_*\ll r_C$, the resolved crossover obeys
\[
x_*=(4r_C^2\ellstar)^{1/3}
\left(1+O\!\left(\frac{x_*^2}{r_C^2}\right)\right)
=\left(\frac{4G\amu^2r_C^2}{\lambda\hbar}\right)^{1/3}
\left(1+O\!\left(\frac{x_*^2}{r_C^2}\right)\right).
\]
Thus, in this regime, $x_*$ scales as $r_C^{2/3}\lambda^{-1/3}$. For the GRW values above, $x_*/r_C\approx1.91\times10^{-2}$, so the small-separation approximation is self-consistent.
\end{corollary}

\begin{proof}
For $x\ll r_C$,
\[
x\left(1-e^{-x^2/(4r_C^2)}\right)
=\frac{x^3}{4r_C^2}\left(1+O\!\left(\frac{x^2}{r_C^2}\right)\right).
\]
Setting this expression equal to $\ellstar$ gives the stated scaling.
\end{proof}

\begin{theorem}[Radius-regularized crossover surface]
Assume $\dx<2R$, so that $d_{\mathrm{eff}}=2R$. Then
\[
\Lcsl(P)\ge\Ldp(P)
\quad\Longleftrightarrow\quad
R\ge R_*(\dx):=\frac{\ellstar}{2\Krc(\dx)}.
\]
For $\dx\ll r_C$,
\[
R_*(\dx)\sim\frac{2r_C^2\ellstar}{\dx^2}.
\]
\end{theorem}

\begin{proof}
In the radius-regularized regime,
\[
\Xic=\frac{2R\Krc(\dx)}{\ellstar}.
\]
Thus $\Xic\ge1$ is equivalent to $R\ge\ellstar/(2\Krc(\dx))$. The asymptotic follows from $\Krc(\dx)=\dx^2/(4r_C^2)+O(\dx^4)$.
\end{proof}

\begin{remark}
The inequality $\Lcsl>\Ldp$ is not an observability criterion. Both exponents may be much smaller than one, in which case neither predicted effect is measurable. The crossover is only a relative calibration surface. A viable proposal must also satisfy an absolute sensitivity requirement, for example $\max\{\Lcsl,\Ldp\}\gtrsim1$, after ordinary environmental decoherence and instrumental noise have been included.
\end{remark}

Figures~\ref{fig:phase-diagram} and \ref{fig:ratio-threshold} display the point-particle conversion law in complementary forms. Figure~\ref{fig:phase-diagram} shows the full geometry plane and separates the resolved and radius-regularized branches. Figure~\ref{fig:ratio-threshold} instead plots the dimensionless ratio directly, making clear that crossing $\Xic=1$ compares the two exponents but does not certify that either is experimentally detectable. The protocol markers are illustrative parameter points; the boundary curves are analytic.

\begin{figure*}[tbp]
\centering
\includegraphics[width=0.94\textwidth]{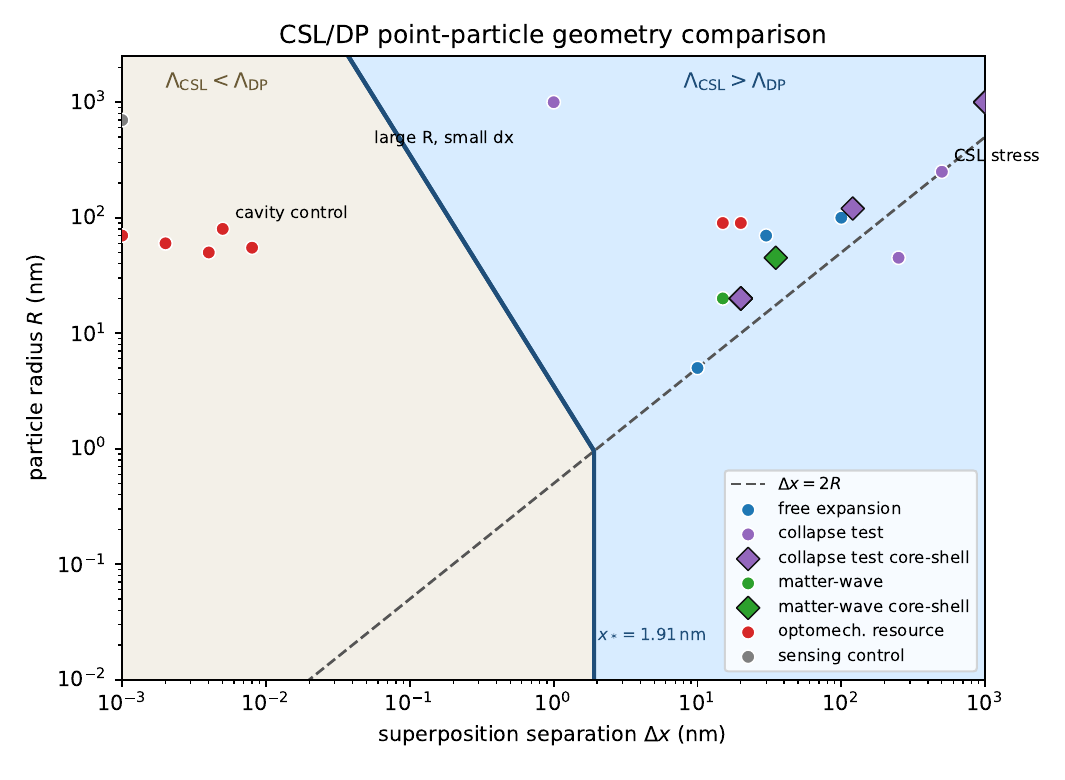}
\caption{Point-particle geometry diagram for the GRW reference parameters. The blue region satisfies $\Lcsl>\Ldp$, whereas the beige region satisfies $\Lcsl<\Ldp$. The solid blue curve is the exact threshold $d_{\mathrm{eff}}\Krc(\dx)=\ellstar$. The dashed line marks only the change of regularization at $\dx=2R$: below it, $d_{\mathrm{eff}}=\dx$; above it, $d_{\mathrm{eff}}=2R$. The short vertical segment marks the resolved crossover $x_*\approx1.91\,\mathrm{nm}$.}
\label{fig:phase-diagram}
\end{figure*}

\begin{figure*}[tbp]
\centering
\includegraphics[width=0.94\textwidth]{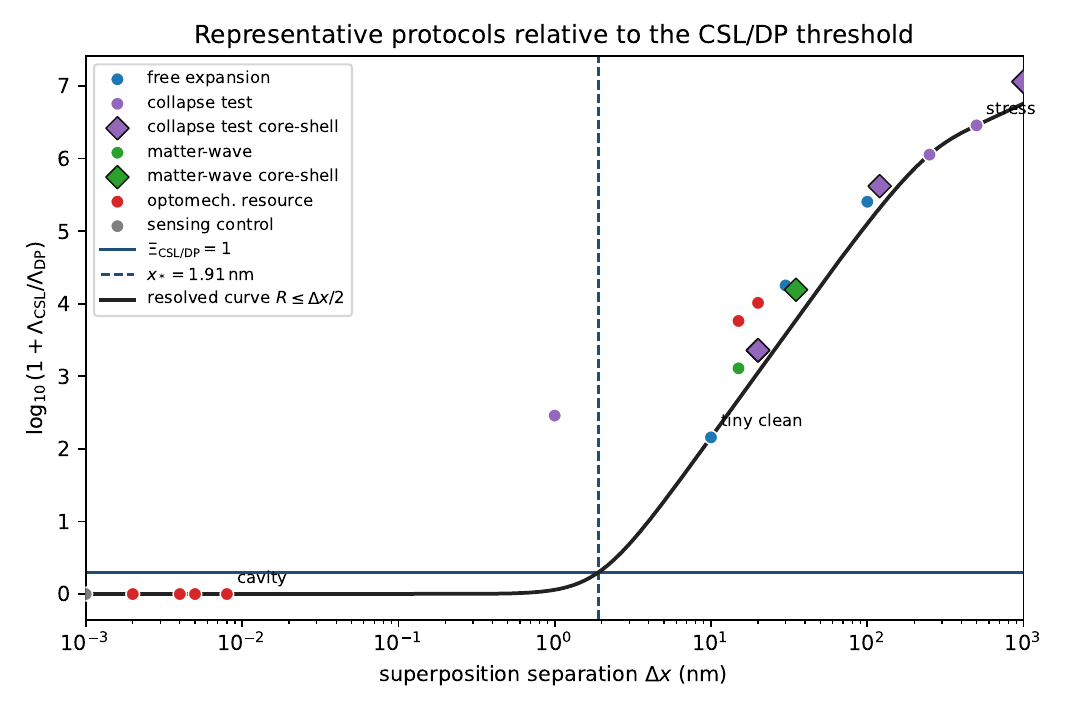}
\caption{The point-particle ratio $\Xic$ for the illustrative protocol set. The horizontal line $\Xic=1$ is equivalent to $\Lcsl=\Ldp$. The black curve is the analytic resolved-regime relation, valid for $R\le\dx/2$. Markers above or below the line indicate relative ordering only; absolute detectability must be assessed separately.}
\label{fig:ratio-threshold}
\end{figure*}

\section{Finite Spherical Mass Profiles}

The point-particle comparison exposes the cancellation, but extended rigid bodies and mass-density profiles are standard in collapse-model calculations \cite{FerialdiBassi2020}, and the DP proposal is normally expressed through the gravitational self-energy of the difference between two mass distributions \cite{Diosi1987,Diosi1989,Penrose1996,Diosi2022}. Thus the useful question is whether the cancellation survives finite-size form factors and nonuniform radial density.

The normalization below separates scale from shape. The total mass is $m$, while $\varrho$ records only how that mass is distributed inside the ball. Changing the material density scale while keeping the same normalized profile changes $m$, but not $\varrho$.

This distinction is important experimentally. Two particles can have the same radius and total mass but different radial composition, for instance through a dense core, a light coating, or a shell-heavy design. Conversely, the same normalized profile can be made heavier by changing the overall density scale. The theorem below says that the second operation does not affect the CSL/DP ratio, while the first can.

Let $\varrho$ be a rigid spherically symmetric density profile supported in a ball of radius $R$, normalized to total mass one after rescaling:
\[
\int_{|x|\le R}\varrho(x)\,d^3x=1.
\]
The physical mass density is $m\varrho(x)$. Write $F_\varrho(u)$ for its normalized radial Fourier form factor,
\[
F_\varrho(u)
=\int_{|x|\le R}\varrho(x)e^{iu\hat z\cdot x/R}\,d^3x
=4\pi\int_0^R r^2\varrho(r)\frac{\sin(ur/R)}{ur/R}\,dr,
\]
where the last expression uses spherical symmetry and the quotient is interpreted continuously at $u=0$. In particular $F_\varrho(0)=1$. With $\sinc z=\sin z/z$, define the finite-size CSL kernel
\[
K_{\mathrm{CSL}}^\varrho(\dx,R;r_C)
=\frac{4}{\sqrt\pi}\int_0^\infty q^2e^{-q^2}
\left|F_\varrho\!\left(q\frac{R}{r_C}\right)\right|^2
\bigl(1-\sinc(q\dx/r_C)\bigr)\,dq,
\]
with the continuous value at $\dx=0$. Define the DP shape kernel by
\[
K_{\mathrm{DP}}^\varrho(\dx/R)
=\frac{2}{\pi}\int_0^\infty |F_\varrho(u)|^2
\bigl(1-\sinc(u\dx/R)\bigr)\,du.
\]

These normalizations are chosen so that the CSL kernel reduces to $\Krc(\dx)$ in the point-particle limit and the DP kernel has the standard homogeneous-sphere form stated below. Other regularization conventions change the kernels and the resulting numerical crossover, but the conversion argument applies unchanged once the two kernels are specified consistently.

The form factor $F_\varrho$ is the object through which internal mass placement enters both theories. In the CSL kernel it is sampled with a Gaussian weight set by $r_C$, so mass features at wavelengths much shorter than $r_C$ are suppressed. In the DP kernel it enters the Newtonian self-energy of the difference between two displaced copies of the same mass distribution. Thus the two kernels probe the same radial profile through different filters.

For such a profile, define the two displaced normalized densities by
\[
\varrho_\pm(x)=\varrho(x\mp\dx\hat z/2).
\]
With the DP convention
\[
E_G^\varrho
=\frac{Gm^2}{2}\int_{\mathbb R^3}\!\int_{\mathbb R^3}
\frac{(\varrho_+(x)-\varrho_-(x))(\varrho_+(y)-\varrho_-(y))}{|x-y|}
\,d^3x\,d^3y,
\]
the Fourier representation of the Newton kernel gives
\[
E_G^\varrho=\frac{Gm^2}{R}K_{\mathrm{DP}}^\varrho(\dx/R),
\qquad
\Lambda_{\mathrm{DP}}^\varrho=\frac{E_G^\varrho\tau}{\hbar}.
\]
The corresponding finite-size CSL contrast-loss exponent is
\[
\Lambda_{\mathrm{CSL}}^\varrho
=\lambda\left(\frac{m}{\amu}\right)^2\tau
K_{\mathrm{CSL}}^\varrho(\dx,R;r_C).
\]

The displacement direction is immaterial because the profile is spherical. For nonspherical bodies, the same Fourier strategy can be used, but the ratio will generally depend on orientation and on the full three-dimensional form factor.

\begin{theorem}[Finite spherical-profile conversion]
Assume $\dx>0$ and $K_{\mathrm{DP}}^\varrho(\dx/R)>0$. For every rigid spherically symmetric profile $\varrho$ in the model above,
\[
\frac{\Lambda_{\mathrm{CSL}}^\varrho}{\Lambda_{\mathrm{DP}}^\varrho}
=\frac{\lambda\hbar R}{G\amu^2}
\frac{K_{\mathrm{CSL}}^\varrho(\dx,R;r_C)}
{K_{\mathrm{DP}}^\varrho(\dx/R)}
=\frac{R K_{\mathrm{CSL}}^\varrho(\dx,R;r_C)}
{\ellstar K_{\mathrm{DP}}^\varrho(\dx/R)}.
\]
In particular, for a fixed normalized spherical profile, the ratio is independent of the total particle mass $m$, the overall density scale, and the interrogation time $\tau$.
\end{theorem}

\begin{proof}
Both sides have the same mass-squared/time-linear scaling:
\[
\Lambda_{\mathrm{CSL}}^\varrho
=\lambda\frac{m^2}{\amu^2}\tau K_{\mathrm{CSL}}^\varrho(\dx,R;r_C),
\qquad
\Lambda_{\mathrm{DP}}^\varrho
=\frac{Gm^2\tau}{\hbar R}K_{\mathrm{DP}}^\varrho(\dx/R).
\]
Dividing cancels $m^2$ and $\tau$. Substituting $\ellstar=G\amu^2/(\lambda\hbar)$ gives the equivalent geometry-factor form.
\end{proof}

\begin{corollary}[Spherical-profile dominance surface]
For rigid spherically symmetric profiles with $\dx>0$,
\[
\Lambda_{\mathrm{CSL}}^\varrho\ge\Lambda_{\mathrm{DP}}^\varrho
\quad\Longleftrightarrow\quad
R K_{\mathrm{CSL}}^\varrho(\dx,R;r_C)
\ge\ellstar K_{\mathrm{DP}}^\varrho(\dx/R).
\]
\end{corollary}

\begin{corollary}[Homogeneous sphere]
\label{cor:homogeneous-sphere}
For a rigid homogeneous sphere,
\[
F(u)=
\begin{cases}
1,&u=0,\\[1mm]
3\dfrac{\sin u-u\cos u}{u^3},&u\ne0,
\end{cases}
\]
and the DP shape kernel reduces to
\[
K_{\mathrm{DP}}^{\mathrm{sph}}(a)=
\begin{cases}
\dfrac{a^2}{2}-\dfrac{3a^3}{16}+\dfrac{a^5}{160},&0\le a\le2,\\[2mm]
\dfrac65-\dfrac1a,&a\ge2.
\end{cases}
\]
Thus
\[
\frac{\Lambda_{\mathrm{CSL}}^{\mathrm{sph}}}{\Lambda_{\mathrm{DP}}^{\mathrm{sph}}}
=\frac{R K_{\mathrm{CSL}}^{\mathrm{sph}}(\dx,R;r_C)}
{\ellstar K_{\mathrm{DP}}^{\mathrm{sph}}(\dx/R)}.
\]
\end{corollary}

\begin{proof}
The normalized density is $\varrho=3/(4\pi R^3)$ on the ball of radius $R$. Direct integration of its radial Fourier transform gives the stated form factor. Substitution into the DP kernel, followed by elementary evaluation in the overlapping regimes $0\le a\le2$ and $a\ge2$, gives the displayed piecewise polynomial. The ratio then follows from the finite spherical-profile conversion theorem.
\end{proof}

\begin{remark}
Finite-size corrections do not, by themselves, reintroduce mass or time dependence. They replace the point-particle factor $d_{\mathrm{eff}}\Krc(\dx)$ by the spherical geometry factor
\[
R\frac{K_{\mathrm{CSL}}^\varrho(\dx,R;r_C)}{K_{\mathrm{DP}}^\varrho(\dx/R)}.
\]
The remaining finite-size question is how strongly the normalized density profile changes this geometry factor.
\end{remark}

The two kernels have a useful physical reading. The CSL kernel is controlled by the localization length $r_C$: it weights the mass form factor at wavelengths comparable with $r_C$ and the branch separation. The DP kernel is controlled by the self-energy of the difference between two displaced mass distributions, naturally measured on the radius scale $R$. Hence, the ratio is sensitive to
\[
s=\frac{R}{r_C},
\qquad
a=\frac{\dx}{R}.
\]
The first parameter asks whether the particle is small or large compared with the CSL localization length; the second asks whether the branches are close or well separated compared with the particle size.

Figure~\ref{fig:finite-identity} checks the finite-profile formula by evaluating its two sides independently on the illustrative protocol set. Its purpose is not to infer the identity numerically, but to show how heterogeneous radii, separations, and radial profiles collapse to the single geometry quotient predicted by the theorem.

\begin{figure*}[tbp]
\centering
\includegraphics[width=0.90\textwidth]{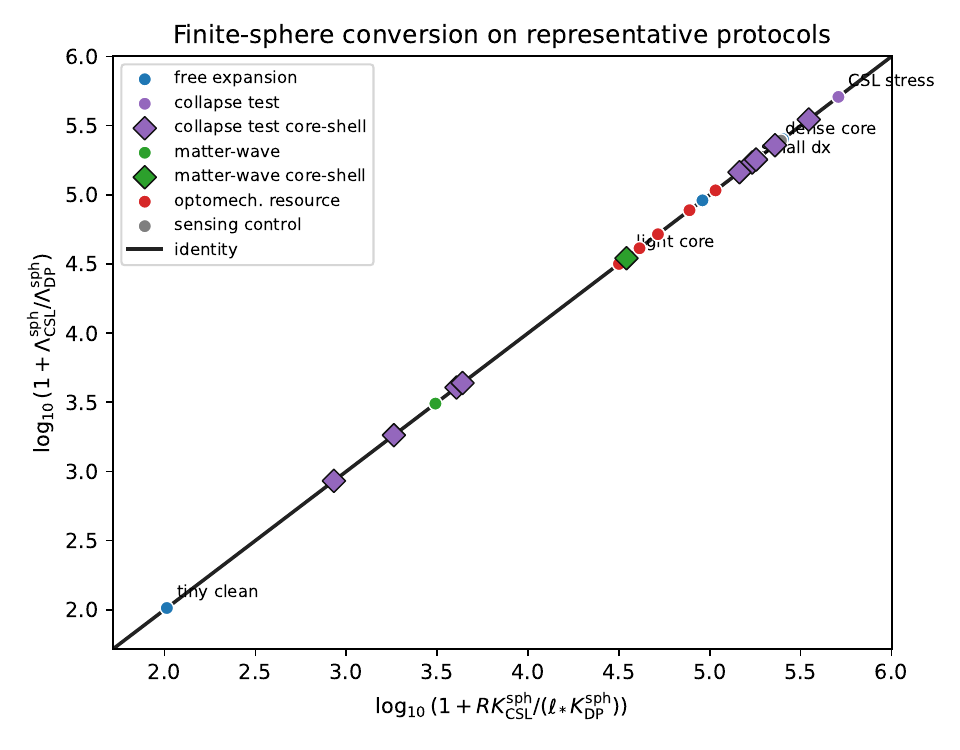}
\caption{Numerical evaluation of the finite spherical-profile identity for the illustrative protocol set. Diamond markers denote core-shell profiles. The horizontal coordinate is the geometry factor $R K_{\mathrm{CSL}}^\varrho/(\ellstar K_{\mathrm{DP}}^\varrho)$; the vertical coordinate is the ratio computed independently from the two finite-size exponents. Equality with the diagonal is the analytic content of the conversion theorem, rather than a fitted relation.}
\label{fig:finite-identity}
\end{figure*}

\section{A Shape-Engineering Conjecture}
\label{sec:shape_engineering_conjecture}
The spherical-profile theorem removes mass, overall density scale, and interrogation time from the ratio, but it leaves the normalized radial profile. We now isolate that dependence and ask whether internal mass placement can systematically favor one side of the CSL/DP comparison. Write
\[
Q_\varrho(s,a):=
\frac{K_{\mathrm{CSL}}^\varrho(aR,R;r_C)}{K_{\mathrm{DP}}^\varrho(a)},
\qquad
s=\frac{R}{r_C},
\qquad
a=\frac{\dx}{R}.
\]
The dimensionless quotient $Q_\varrho$ contains all dependence on the profile at fixed $(s,a)$. Radial mass engineering changes the CSL/DP ratio precisely when it changes this quotient.

Compare it with the homogeneous-sphere value
\[
Q_{\mathrm{hom}}(s,a)=
\frac{K_{\mathrm{CSL}}^{\mathrm{sph}}(aR,R;r_C)}{K_{\mathrm{DP}}^{\mathrm{sph}}(a)}
\]
through the profile amplification factor
\[
\mathcal A_\varrho(s,a)
=\frac{Q_\varrho(s,a)}{Q_{\mathrm{hom}}(s,a)}.
\]

Thus $\mathcal A_\varrho>1$ means that the profile increases the CSL/DP geometry factor relative to a homogeneous sphere with the same $R$ and $\dx/R$, whereas $\mathcal A_\varrho<1$ means that it suppresses the factor.

The homogeneous sphere is the natural baseline because it removes compositional design from the comparison. A value $\mathcal A_\varrho\ne1$ therefore records a profile effect rather than a change in total mass, radius, or separation ratio. Values above one shift the relative exponent toward CSL compared with the homogeneous sphere; values below one shift it toward DP.

For a two-layer core-shell profile, let $f\in(0,1)$ be the core radius divided by $R$, and let $\eta>0$ be the core density divided by the shell density before normalization. Thus $\eta>1$ describes a denser core, while $\eta<1$ biases mass toward the shell. If
\[
F_{\mathrm{sph}}(u)=3\frac{\sin u-u\cos u}{u^3},
\qquad F_{\mathrm{sph}}(0)=1,
\]
then direct integration gives the normalized two-layer form factor
\begin{equation}
F_{f,\eta}(u)
=\frac{F_{\mathrm{sph}}(u)+(\eta-1)f^3F_{\mathrm{sph}}(fu)}
{1+(\eta-1)f^3}.
\label{eq:core-shell-form-factor}
\end{equation}
The corresponding core mass fraction is
\[
\frac{m_{\mathrm{core}}}{m}
=\frac{\eta f^3}{1+(\eta-1)f^3}.
\]
Equations~\eqref{eq:core-shell-form-factor} and the two kernel definitions make every value below directly reproducible.

Direct quadrature suggests a profile inversion. When $R\ll r_C$, the CSL filter has little sensitivity to internal radial structure, whereas the DP self-energy remains sensitive to mass placement. A dense core can then give $\mathcal A_{f,\eta}<1$. As $R/r_C$ increases, the Gaussian CSL filter resolves progressively shorter radial scales, and the same dense-core profile can cross to $\mathcal A_{f,\eta}>1$. Shell-weighted profiles remain above one in the displayed examples, but the calculations do not establish that they are uniformly optimal.

The evidence in this section comes from direct quadrature of the two kernels for the same normalized form factor, rather than from fitting a parametric model. Figure~\ref{fig:finite-identity} checks the ratio implementation, while Figure~\ref{fig:core-shell} and Table~\ref{tab:two-layer-scan} record the profile dependence that motivates the conjecture.

\begin{conjecture}[Core-shell inversion]
Fix a separation ratio $a>0$ and a finite density-contrast bound $M>1$. Let $\varrho_{f,\eta}$ range over normalized two-layer spherical profiles with core radius $fR$, $0<f<1$, and core-to-shell density ratio $\eta\in[M^{-1},M]$. A maximizer of
\[
\mathcal A_{f,\eta}(s,a)
\]
over these profiles is attained on an extreme-contrast branch $\eta=M$ or $\eta=M^{-1}$. As $s=R/r_C$ increases, the maximizing branch changes from shell weighted to core weighted, and for fixed $a$ and $M$ the two extreme branches exchange order only finitely many times. In particular, sufficiently concentrated dense-core profiles cross from $\mathcal A_{f,\eta}<1$ for $s\ll1$ to $\mathcal A_{f,\eta}>1$ for $s\gg1$.
\end{conjecture}

The finite-crossing assertion is intentionally weaker than uniqueness of the transition: the available quadrature does not prove monotonicity of the Gaussian-filtered form-factor integral. A proof would convert the cancellation theorem into a profile-selection principle. A counterexample would instead identify radial designs more subtle than the two extreme branches.

\begin{figure*}[tbp]
\centering
\includegraphics[width=0.94\textwidth]{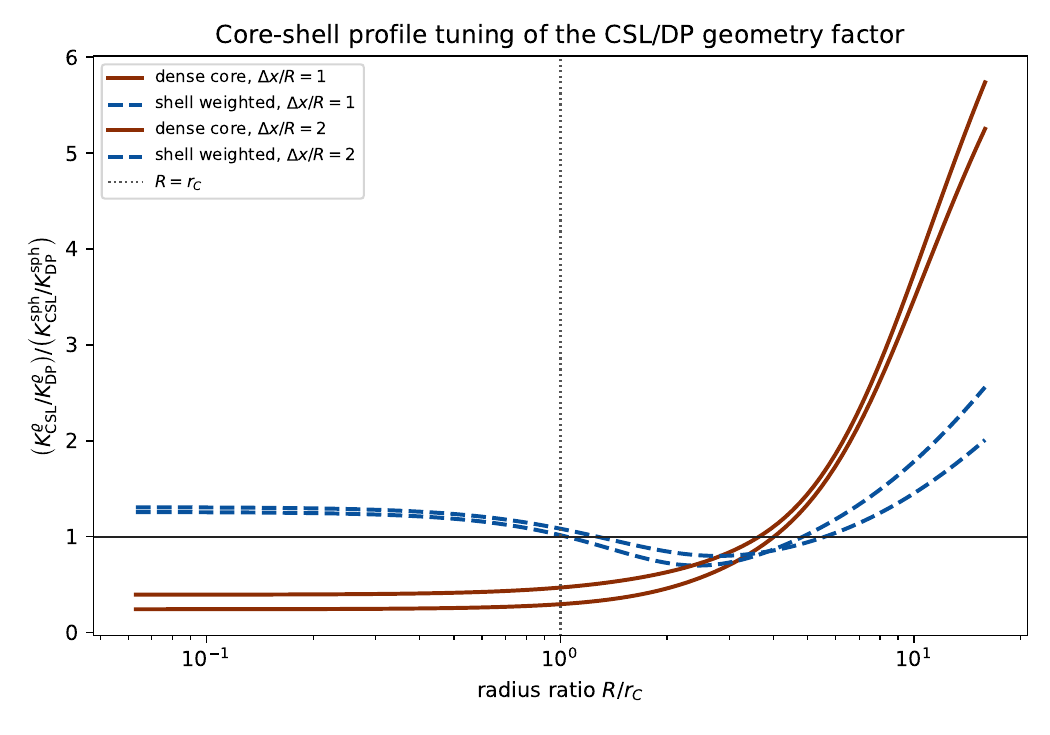}
\caption{Core-shell tuning of the spherical CSL/DP geometry factor. The plotted quantity is $\mathcal A_\varrho=Q_\varrho/Q_{\mathrm{hom}}$ at fixed $R/r_C$ and $\dx/R$. Shell-weighted profiles enhance the ratio in the displayed small-radius regime, while the dense-core profile crosses from suppression to amplification as $R/r_C$ increases.}
\label{fig:core-shell}
\end{figure*}

\begin{table*}[tbp]
\centering
\begin{tabular}{lrrrrr}
\toprule
profile & $R/r_C$ & $f$ & $\eta$ & $m_{\mathrm{core}}/m$ & $\mathcal A_{f,\eta}(s,1)$\\
\midrule
dense core  & $0.2$ & $0.2$ & $100$  & $0.446$ & $0.244$\\
dense core  & $10$  & $0.2$ & $100$  & $0.446$ & $3.480$\\
shell weighted & $0.2$ & $0.9$ & $0.02$ & $0.051$ & $1.244$\\
shell weighted & $10$  & $0.9$ & $0.02$ & $0.051$ & $1.785$\\
\bottomrule
\end{tabular}
\caption{Representative two-layer values at $a=\dx/R=1$, obtained by direct quadrature. The same dense-core profile suppresses the CSL/DP geometry factor at $R/r_C=0.2$ and amplifies it at $R/r_C=10$. The total mass does not enter these values.}
\label{tab:two-layer-scan}
\end{table*}

\section{Experimental Implications and Open Problems}
\label{sec:open_problems}
The exact conversion laws reduce the theoretical comparison to geometry, but an experiment cannot vary geometry independently of preparation fidelity, environmental decoherence, heating, and readout. The mathematically largest value of $\Xic$ need not occur in an experimentally accessible region, and maximizing the ratio alone can select protocols for which both exponents are negligible. A useful optimization must therefore retain both relative ordering and absolute sensitivity.

\begin{problem}[Accessible geometry frontier]
For a specified levitated-mass architecture, let $\mathcal A$ be the set of protocols satisfying quantitative preparation, coherence-survival, heating, and readout constraints. Determine the attainable region
\[
\left\{\bigl(\Lambda_{\mathrm{CSL}}(P),
\Lambda_{\mathrm{DP}}(P),\Xic(P)\bigr):P\in\mathcal A\right\}
\]
and its Pareto frontier. In particular, determine whether $\mathcal A$ contains protocols on both sides of $\Xic=1$ for which at least one exponent is experimentally resolvable.
\end{problem}

The answer necessarily depends on the chosen CSL parameters and on the experimentally admissible ranges of $R$, $\dx$, $\tau$, and radial profile. The illustrative markers in Figures~\ref{fig:phase-diagram}--\ref{fig:finite-identity} are not a substitute for constructing $\mathcal A$ from a specific apparatus.

The random-unitary theorem also changes what constitutes a decisive experiment. Measuring only the decay of an off-diagonal density-matrix element reconstructs an unconditional kernel, but does not determine how that kernel is unravelled in individual runs. The heating rate in \eqref{eq:momentum-diffusion} supplies one consistency check, not a discriminator, because it is fixed by the same random-kick generator. Discrimination must instead use information that is absent from the ensemble kernel, such as accessible environmental records, trajectory-conditioned statistics, higher-order temporal correlations, or many-body amplification laws.

\begin{problem}[Discriminating unravellings]
Among dynamical models that reproduce the same point-particle separation kernel $\Krc$, determine a minimal experimentally accessible collection of observables that distinguishes state-independent random-unitary noise from objective localization. In particular, determine whether trajectory-conditioned records, multi-time correlations, or controlled many-body mass scaling can provide a discriminator that is robust to ordinary environmental decoherence.
\end{problem}

The finite spherical-profile theorem resolves the first geometry-reduction question: fixed normalized spherical profiles preserve the mass/time cancellation. Real nanoparticle protocols, however, may require nonspherical particles, anisotropic traps, internal-temperature models, dissipative collapse variants, or experimentally calibrated decoherence budgets. The next problem is to determine which corrections preserve a geometry-only comparison and which introduce genuinely material-dependent behavior.

\begin{problem}[Nonspherical extensions]
Extend the finite-size model from rigid spherical profiles to nonspherical particles, anisotropic mass distributions, internal-temperature corrections, and dissipative collapse variants. Determine which classes still admit a CSL/DP comparison of the form
\[
\frac{\Lambda_{\mathrm{CSL}}^{\mathrm{shape}}}
{\Lambda_{\mathrm{DP}}^{\mathrm{shape}}}
=\frac{1}{\ellstar}\,
\mathcal G(\text{dimensionless geometry})
\]
and which classes force material, temperature, or internal-structure parameters to remain in the ratio.
\end{problem}

\section{Discussion}

We have separated a comparison that is often embedded in a much larger experimental design problem. The analysis does not model trap engineering, readout noise, gas collisions, blackbody emission, feedback, or release-and-recapture dynamics. Instead, it determines which parameters control the relative CSL and DP exponents after the geometry of an idealized spatial superposition has been specified.

The random-unitary realization places a necessary limit on the interpretation of that comparison. The point-particle CSL kernel is an experimentally meaningful law for ensemble coherence, but it is not an ontologically complete description. A measured visibility decay consistent with $\Krc$ would constrain the unconditional generator; it would not, on that evidence alone, establish spontaneous localization or exclude a classical stochastic Hamiltonian. This does not weaken searches for CSL-like decoherence. It clarifies the additional burden of a collapse test: the experiment must either rule out relevant random-unitary alternatives or probe structure that those alternatives do not share.

For the spherical models considered here, the answer is geometric. In the point-particle approximation, mass and interrogation time cancel exactly from the CSL/DP ratio. They remain essential experimentally because they increase the absolute size of both exponents, but they do not decide which exponent is larger. For the GRW reference parameters, the resolved crossover $x_*\approx1.91\,\mathrm{nm}$ is consequently independent of particle mass and interrogation time. This value is a feature of the stated point-particle proxy and parameter choice, not a universal boundary between CSL and DP phenomenology.

The finite-size theorem shows that this cancellation is not merely an artifact of replacing the particle by a point mass with a cutoff. For any fixed normalized spherical mass profile, the finite CSL kernel and DP self-energy kernel again leave a geometry-only ratio. This matters for nanoparticle design because radius and internal density profile are not just secondary corrections. Once mass and time cancel from the comparison, the remaining freedom is precisely the geometry of the object and the geometry of the superposition.

The core-shell conjecture concerns only this residual geometric freedom. It does not assert that a shell-heavy or dense-core nanoparticle is easier to fabricate, cool, prepare in superposition, or read out. It predicts that, at fixed total mass, outer radius, and separation ratio, radial mass placement can change the relative CSL/DP exponent and that the favorable extreme profile changes as $R/r_C$ increases. Establishing or refuting this statement would determine whether internal mass placement is a genuine optimization variable or merely a small correction to homogeneous-sphere estimates.

Thus the contribution is best viewed as a reduction of a multidimensional design comparison. It separates absolute sensitivity, which still requires mass, time, isolation, and platform-specific noise modeling, from relative CSL/DP dominance, which in these spherical models is controlled by geometry, and from identification of the underlying stochastic dynamics. The natural next step is to combine the geometry factors proved here with realistic protocol constraints and an explicit unravelling-discrimination strategy for a specific experimental architecture.

\section*{Data Availability}

No original experimental data were created or analyzed in this study. The
numerical values shown in the figures and table are obtained by direct
evaluation of the equations and parameter values given in the article. The
custom plotting script has not been moved to a public repository, but it can be
made available upon reasonable request.

\bibliography{Davila_Milburn_2026a_PRD}

\end{document}